\documentclass[lettersize,journal]{IEEEtran}
\usepackage{amsmath,amsfonts}
\usepackage{algorithmic}
\usepackage{algorithm}
\usepackage{array}
\usepackage[caption=false,font=footnotesize]{subfig}
\usepackage{textcomp}
\usepackage{stfloats}
\usepackage{url}
\usepackage{verbatim}
\usepackage{graphicx}
\usepackage{cite}

\usepackage{mathtools}
\usepackage{siunitx}
\usepackage{bbm}
\usepackage{booktabs}

\newcommand{\frealp}{\mathbb{R}_{\geq 0}}
\newcommand{\fnatural}{\mathbb{N}}

\newcommand{\lEsNO}{E_\textnormal{s} / N_0}
\newcommand{\lEbNO}{E_\textnormal{b} / N_0}
\newcommand{\ztrafo}[1]{\mathcal{Z}\{#1\}(z)}

\newcommand{\symcell}{g}
\newcommand{\bitcell}{b}
\newcommand{\numbitpermodsymb}{\ell}

\newcommand{\rb}{\mathbb{R}}

\newcommand{\cc}{\mathcal{C}}

\newcommand{\ind}{\mathbbm{1}}

\makeatletter
\def\smallunderbrace#1{\mathop{\vtop{\m@th\ialign{##\crcr
   $\hfil\displaystyle{#1}\hfil$\crcr
   \noalign{\kern3\p@\nointerlineskip}%
   \tiny\upbracefill\crcr\noalign{\kern3\p@}}}}\limits}
\makeatother

\newcommand{\ccc}{\cc_\text{c}}
\newcommand{\cccb}{\cc_\text{c,b}}

\DeclareMathOperator{\argmax}{argmax}

\usepackage{bm}

\usepackage{amsthm}
\newtheorem{theorem}{Theorem}[section]

\usepackage{enumitem}

\usepackage{acro}

\usepackage{nicefrac}

\DeclareAcronym{TPC}{
  short = TPC,
  long = turbo product code,
  short-plural = s
}
\DeclareAcronym{RLC}{
  short = RLC,
  long = random linear code,
  short-plural = s
}
\DeclareAcronym{MDS}{
  short = MDS,
  long = maximum distance separable,
  short-plural = s
}
\DeclareAcronym{CRC}{
  short = CRC,
  long = cyclic redundancy check,
  short-plural = s
}
\DeclareAcronym{BCH}{
  short = BCH,
  long = Bose–Chaudhuri–Hocquenghem,
  short-plural = s
}
\DeclareAcronym{BLER}{
  short = BLER,
  long = block error rate,
  short-plural = s
}
\DeclareAcronym{BI-AWGN}{
  short = BI-AWGN,
  long = binary-input additive white Gaussian noise,
  short-plural = s
}
\DeclareAcronym{SNR}{
  short = SNR,
  long = signal-to-noise ratio,
  short-plural = s
}
\DeclareAcronym{LLR}{
  short = LLR,
  long = log likelihood ratio,
  short-plural = s
}
\DeclareAcronym{SISO}{
  short = SISO,
  long = soft-input soft-output,
  short-plural = s
}

\DeclareAcronym{SO}{
    short = SO,
    long = soft-output
}

\DeclareAcronym{GRAND}{
  short = GRAND,
  long = guessing random additive noise decoding,
  short-plural = s
}

\DeclareAcronym{RLCode}{
  short = RLC,
  long = random linear code,
  short-plural = s
}

\DeclareAcronym{SOGRAND}{
  short = SOGRAND,
  long = soft-output GRAND,
  short-plural = s
}

\DeclareAcronym{LDPC}{
  short = LDPC,
  long = low-density parity-check,
  short-plural = s
}

\DeclareAcronym{VN}{
  short = VN,
  long = variable node,
  short-plural = s
}

\DeclareAcronym{CN}{
  short = CN,
  long = check node,
  short-plural = s
}

\DeclareAcronym{RS}{
  short = RS,
  long = Reed-Solomon,
  short-plural = s
}

\DeclareAcronym{AWGN}{
  short = AWGN,
  long = additive white Gaussian noise,
  short-plural = s
}

\DeclareAcronym{ISI}{
  short = ISI,
  long = intersymbol interference,
  short-plural = s
}

\DeclareAcronym{SDD}{
    short = SDD,
    long = separate detection and decoding,
    short-plural = s
}

\DeclareAcronym{JDD}{
    short = JDD,
    long = joint detection and decoding,
    short-plural = s
}

\DeclareAcronym{APP}{
    short = APP,
    long = a posteriori probability,
    short-plural = s,
    long-plural-form = a posteriori probabilities
}

\DeclareAcronym{ML}{
    short = ML,
    long = maximum likelihood,
    short-plural = s
}

\DeclareAcronym{GMI}{
    short = GMI,
    long = generalized mutual information,
    short-plural = s
}

\DeclareAcronym{GCD}{
    short = GCD,
    long = guessing codeword decoding,
    short-plural = s
}

\begin{document}

\title{Group Probability Decoding of Turbo Product Codes over Higher-Order Fields}

\author{Lukas Rapp,~\IEEEmembership{Graduate Student Member,~IEEE,}, Muriel Médard,~\IEEEmembership{Fellow,~IEEE}, and Ken R. Duffy,~\IEEEmembership{Senior~Member,~IEEE}
\thanks{This paper was presented in part at 2025 CISS~\cite{rappSOGRANDDecodingNonbinary}.}%
\thanks{L. Rapp and M. Médard are with the Massachusetts Institute of Technology Network Coding \& Reliable Communications Group (e-mails: {rappl, medard}@mit.edu).}%
\thanks{K. R. Duffy is with the Northeastern University Engineering Probability Information \& Communications Laboratory (e-mail: k.duffy@northeastern.edu).}%
\thanks{This work was supported by the Defense Advanced Research Projects Agency (DARPA) under Grant HR00112120008.}
}

\maketitle

\begin{abstract}
Binary \acp{TPC} are powerful error-correcting codes constructed from short component codes. Traditionally, turbo product decoding passes \acp{LLR} between the component decoders, inherently losing information when bit correlation exists. Such correlation can arise exogenously from sources like \acl{ISI} and endogenously during component code decoding. To preserve these correlations and improve performance, we propose turbo product decoding based on group probabilities. We theoretically predict mutual information and \ac{SNR} gains of group over bit-probability decoding. To translate these theoretical insights to practice, we revisit non-binary \acp{TPC} that naturally support group-probability decoding. We show that any component list decoder that takes group probabilities as input and outputs block-wise soft-output can partially preserve bit correlation, which we demonstrate with symbol-level ORBGRAND combined with \ac{SOGRAND}.
Our results demonstrate that group-probability-based turbo product decoding achieves \ac{SNR} gains of up to 0.3 dB for endogenous correlation and 0.7 dB for exogenous correlation, compared to bit-probability decoding. 
\end{abstract}


\acresetall

\section{Introduction}
An effective way to construct long channel codes is to concatenate short component codes, a principle first demonstrated by Elias with product codes~\cite{eliasErrorfreeCoding1954a} and later adopted in \ac{LDPC} codes by Gallager~\cite{gallagerLowdensityParitycheckCodes1962a} with single parity check component codes. Following the advent of turbo codes~\cite{berrouShannonLimitErrorcorrecting1993}, Pyndiah~\cite{pyndiahNearoptimumDecodingProduct1998} proposed soft-decision turbo decoding for product codes, which are then referred to as \acp{TPC}~\cite{mukhtarTurboProductCodes2016}.
The core element of Pyndiah's algorithm is the \ac{SISO} component decoder that estimates \ac{SO} information from a codeword list generated by Chase decoding~\cite{chaseClassAlgorithmsDecoding1972}. Recently, interest in list decoding has been renewed by CRC-assisted successive cancellation list (CA-SCL) decoding~\cite{talListDecodingPolar2015} for Polar Codes~\cite{arikanChannelPolarizationMethod2009}.
Recently, \ac{SOGRAND}, an \ac{SISO} decoders with improved \ac{SO}, has been proposed, resulting in enhanced \ac{TPC} performance that is competitive with \ac{LDPC} codes~\cite{yuanSoftoutputGRANDLong2023}. By viewing the operations of other decoders through the lens of \ac{GRAND}, the ability to provide accurate block-wise \ac{SO} has been further extended to SO-GCD~\cite{duffySoftOutputGuessingCodeword2025}, SO-SCL~\cite{yuanSoftOutputSuccessiveCancellation2025}, and SOCS~\cite{Janz25soft}.

In Pyndiah's algorithm and the following improvements, the component decoder generates a codeword list and assigns each codeword an \ac{APP} based on the input soft information. These \acp{APP} form a probability distribution of point masses in \(\{0, 1\}^n\), where the codebook \(\mathcal{C}\) introduces bit correlations. However, existing \ac{TPC} algorithms compute and exchange only bit-wise marginals of the codeword \acp{APP}, effectively transforming the distribution of point masses into a less informative product distribution.
For instance, consider the decoding of a repetition code with \ac{APP} estimates
\begin{equation}\label{eqn:example-repetition}
    \begin{aligned}
    c_1^n &= (0\hspace{0.2em}00\hspace{0.2em}0000000000\hspace{0.2em}0\hspace{0.2em}00), \quad &\text{with \(P(c_1^{n}) = 0.5\)},\\ 
    c_2^n &= (1\underbracket[.7pt]{11}_{\mathclap{\substack{P(00) = P(11)=0.5\\ P(01) = P(10)=0.0}}} 1111111111\underbracket[.7pt]{1}_{\mathclap{\substack{P(0) = 0.5\\P(1) = 0.5}}} 11), \quad &\text{with \(P(c_2^n) = 0.5\)}. 
    \end{aligned}
\end{equation}
Consequently, the probability of each bit being \(0\) is \(\nicefrac{1}{2}\) and the resulting product distribution assigns equal probability to all sequences in \(\{0, 1\}^{16}\) discarding any correlation between the bits. This correlation can be highly informative for the next component decoder as it reduces the \ac{APP} of incorrect codewords, promising improved decoding performance and convergence in fewer iterations. In addition to this decoder-induced \emph{endogenous correlation}, external effects like constellation bit mapping and \ac{ISI} can introduce \emph{exogenous correlation}, which the decoder can leverage.

Recently, Duffy et al.~\cite{duffyUsingChannelCorrelation2023a} proposed ORBGRAND-AI, a practical decoder that leverages exogenous bit correlation for short block codes. To achieve this, it computes the joint \acp{APP} over consecutive bit groups and carries out guesswork informed by these group \acp{APP}~\cite{anSoftDecodingSoft2022a}. For example, if we consider groups of two bits in \eqref{eqn:example-repetition}, we see that the sequences \(01\) and \(10\) have probability \(0\), which partially preserves the bit correlation introduced by the codebook. Inspired by this idea, we theoretically analyze to what extent group probabilities preserve bit correlation and improve decoding performance. Specifically, we compare the mutual information between the channel output and the group \acp{APP} with that between the channel output and bit \acp{APP}. This allows us to predict the decoding gains of ORBGRAND-AI under exogenous correlation for short block codes and to quantify how endogenous and exogenous bit correlation can benefit turbo product decoding.  

Motivated by the theoretical coding gains of group-probability decoding, we propose a practical channel coding scheme based on turbo product decoding to realize the predicted gains for long codes. Exchanging group probabilities between component decoders brings two requirements: first, the bits of each group need to belong to the same row and column; second, the \ac{SISO} component decoder must handle group probabilities as input and output. To fulfill the first requirement, we use non-binary instead of binary product codes, where each cell contains a group of bits.
For group-based \ac{SISO} component decoding, we demonstrate that any component decoder that takes group probabilities as input and outputs a codeword list with block-wise \acp{APP} can be used: group probabilities can be readily obtained by marginalizing the \acp{APP} over bit groups, as shown in \eqref{eqn:example-repetition}. Practical decoders that fulfill these requirements have recently been developed. In this paper, we use symbol-level ORBGRAND~\cite{anSoftDecodingSoft2022a, duffyUsingChannelCorrelation2023a}, based on \ac{GRAND} decoding~\cite{duffyCapacityAchievingGuessingRandom2019, duffyOrderedReliabilityBits2022}, and generate block-wise \ac{SO} using \ac{SOGRAND}~\cite{duffyOrderedReliabilityBits2022}. An alternative example is a symbol-level version of \ac{GCD}~\cite{maGuessingWhatNoise2024a}, where SO-GCD~\cite{duffySoftOutputGuessingCodeword2025} provides the block-wise \ac{SO}. We verify via Monte Carlo simulations that our scheme preserves bit correlation resulting in \ac{SNR} gains of up to \SI{0.3}{\dB} for endogenous correlation and \SI{0.7}{\dB} for exogenous correlation compared to bit-probability decoding.

This paper is structured as follows: Section~\ref{sec:theory} presents an information-theoretical analysis demonstrating how leveraging endogenous and exogenous correlation with group probabilities can improve decoding performance.
In Sec.~\ref{sec:non-binary-tpc}, we revisit non-binary \acp{TPC} as a practical way to achieve these performance gains and outline turbo product decoding with bit and group probabilities.
Section~\ref{sec:res-group-prob} presents numerical simulation results for group-probability decoding of non-binary \acp{TPC}. We use the following notation throughput this paper: \([n] \coloneqq \{1, \dots, n\}\) denotes the integers from \(1\) to \(n\), capitalized italic symbols such as \(\mathcal{A}\) denote sets, \(|\mathcal{A}|\) denotes the cardinality of a set, and \(j\) denotes the imaginary unit.

\section{Related Work}
To the best of our knowledge, previous research on non-binary \acp{TPC} only focuses on decoding with independent bit probabilities, leaving decoding with group probabilities to leverage bit correlation unexplored: turbo product decoding of non-binary \acp{TPC} with bit probabilities was first analyzed by Pyndiah in \cite{aitsabPerformanceReedSolomonBlock1996} to construct new \acp{TPC} with \ac{RS} codes and later improved with more powerful decoders in \cite{sweeneyIterativeSoftdecisionDecoding2000, zhouLowComplexityHighRateReed2007, rappSOGRANDDecodingNonbinary}.

Non-binary turbo decoding with group probabilities was previously studied for other code classes, e.g., non-binary \ac{LDPC} codes \cite{gallagerLowdensityParitycheckCodes1962a, daveyLowDensityParity1998, declercqDecodingAlgorithmsNonbinary2007, voicilaLowComplexityLowMemoryEMS2007, poulliatDesignRegularSub2008, toriyama2267Gb937pJBit2018, ferrazSurveyHighThroughputNonBinary2022} and non-binary convolutional turbo codes \cite{berrouShannonLimitErrorcorrecting1993, berkmannTurboDecodingNonbinary1998, berrouAdvantagesNonbinaryTurbo2001, livaShortTurboCodes2013, klaimiLowcomplexityDecodersNonbinary2018}.
Both non-binary \ac{LDPC} and convolutional codes can outperform their binary counterparts for moderate-to-short block length~\cite{livaShortTurboCodes2013, daveyLowDensityParity1998, poulliatDesignRegularSub2008}.
Additionally, non-binary codes offer the advantage that higher-order constellation symbols can be directly mapped to non-binary codeword symbols, avoiding the costly per-bit demapper~\cite{anSoftDecodingSoft2022a, declercqRegularGF2q, abdmoulehNewApproachOptimise2016}.

In this work, we utilize group probabilities to leverage exogenous bit correlation from \ac{ISI}. Optimal decoding of exogenous correlation can be realized with \ac{JDD}, which selects the codeword with maximal \ac{APP}. Since \ac{JDD} is often too complex, suboptimal solutions like turbo equalization~\cite{douillardIterativeCorrectionIntersymbol1995, koetterTurboEqualization2004}, successive interference cancellation~\cite{pfisterAchievableInformationRates2001, prinzSuccessiveInterferenceCancellation2024, wachsmannMultilevelCodesTheoretical1999}, \ac{SDD} and are used. \ac{SDD} offers low complexity and design flexibility by treating equalization and decoding as independent components, but can experience a significant performance loss.
Passing group instead of bit probabilities from the equalizer to the decoder reduces this loss~\cite{duffyUsingChannelCorrelation2023a} while keeping the detection and decoding separated. Previous work~\cite{mullerCapacityLossDue2004, pfisterAchievableInformationRates2001, arnoldInformationRateBinaryinput2001, kavcicCapacityMarkovSources2001} analyzed the achievable information rates for channels with memory and the capacity loss between \ac{SDD} and \ac{JDD} for bit probabilities~\cite{mullerCapacityLossDue2004}. In this work, we analyze how group-probability-based \ac{SDD} can reduce this loss.

\section{Group-Probability Decoding}\label{sec:theory}
This section presents an information-theoretical analysis of how effectively decoding with group probabilities preserves bit correlation and enhances decoding performance. We differentiate between \emph{exogenous} and \emph{endogenous} correlation between bits, where the \emph{exogenous} correlation is caused by the transmission and the channel itself, for example, via \ac{ISI} and higher order modulation.
The \emph{endogenous} correlation is extracted by the component decoder during iterative decoding.

\subsection{Exogenous Correlation}
\subsubsection{Channel Model}\label{sec:channel-model}
We introduce exogenous correlation via a channel with \ac{ISI}, over which ASK and QAM symbols are transmitted.
A \(2^\numbitpermodsymb\)-ary ASK and QAM constellation, which maps \(\numbitpermodsymb \in \fnatural\) codeword bits \(C_i\) to to \(2^\numbitpermodsymb\) modulation symbols \(X_i\), is defined as:
\begin{align*}
    \mathcal{X}_{2^\numbitpermodsymb,\text{ASK}} 
    &= 
    \left\{
        \pm 1 a, \pm 3 a, \pm 5 a, \dots, \pm (2^{\numbitpermodsymb-1} - 1) a
    \right\},\\
    \mathcal{X}_{2^\numbitpermodsymb,\text{QAM}} 
    &= 
    \left\{
        (x_\text{I} + j x_\text{Q}) a : x_\text{I}, x_\text{Q} \in \mathcal{X}_{2^{\numbitpermodsymb/2},\text{ASK}}
    \right\},
\end{align*}
respectively, where \(a \in \frealp\) is a constant to normalize the average energy of the respective constellation \(\mathcal{X}\) such that \(|\mathcal{X}|^{-1} \sum_{x \in \mathcal{X}} |x|^2 = 1\).
The codeword bits are mapped via a gray coding to ASK and QAM symbols, where, for QAM symbols, a Gray code first maps two groups of \(\numbitpermodsymb / 2\) to two ASK symbols, which are then combined to one complex symbol.

We consider a \ac{ISI} channel~\cite[Sec.~9]{proakisDigital2009} with two taps and \ac{AWGN} over which the modulation symbols \(X\) are transmitted, i.e., the output of the matched filter is
\begin{equation}\label{eqn:channel-model-isi}
    \widetilde{Y}_i = X_i + \rho X_{i-1} + \widetilde{Z}_i,
\end{equation}
where \(\widetilde{Z}_i\) is complex \ac{AWGN} noise with a noise variance of \(\widetilde{\sigma}^2\) per I- and Q-component and \(\rho \in (-1, 1)\).
As shown in Appendix~\ref{sec:gauss-markov-model-derivation}, the intersymbol interference can be removed with a linear equalizer~\cite[Sec.~9.4]{proakisDigital2009}, resulting in
\begin{equation}\label{eqn:channel-model-gm}
    Y_i = X_i + Z_i,
\end{equation}
where \(Z_i\) is complex Gauss-Markov noise as in \cite{duffyUsingChannelCorrelation2023a}: \(Z_i\) has zero mean and the I- and Q-component \(Z_{\text{I},i}\) and \(Z_{\text{Q},i}\) are independent and multivariate normal distributed with auto-covariance
\(
    E\{Z_{\text{I},i_1} Z_{\text{I},i_2}\} 
    = E\{Z_{\text{Q},i_1} Z_{\text{Q},i_2}\} 
    = \sigma^2 \rho^{|i_1-i_2|}
\)
with \(\sigma^2 = \frac{\widetilde{\sigma}^2}{1- \rho^2}\).
This corresponds to an \ac{SNR} of \(\lEsNO = (2 \sigma^2)^{-1} = r \numbitpermodsymb \cdot \lEbNO\), where \(r\) is the rate of the code.

\newcommand{\gstart}{s}
\newcommand{\gend}{e}

\subsubsection{Preprocessing}\label{sec:preprocessing}
To exploit the correlation between symbols \(Y_i\), we divide the equalizer output \(Y^n\) of a codeword \(C^n\) into consecutive, non-overlapping blocks of \(\symcell\) symbols:
\begin{align*}
    Y^n &= \left( Y_1, \dots, Y_\symcell \mid Y_{\symcell+1}, \dots, Y_{2 \symcell} \mid \dots \mid Y_{n - \symcell + 1}, \dots, Y_{n} \right) \\
    &= (Y_{\gstart_1}^{e_1}, \dots, Y_{\gstart_{n}}^{e_{n}}),
\end{align*}
where \(\gstart_i = i \symcell\) and \(\gend_i = (i+1) \symcell-1\). 
For a realization \(y^n\), the preprocessing scheme calculates the \emph{group probabilities}
\begin{equation}\label{eqn:A}
    \text{(gw)} \qquad 
    P_{X_{\gstart}^{\gend} | Y_{\gstart}^{\gend}}\left(
        x_{\gstart}^{\gend} | 
        y_{\gstart}^{\gend}
    \right)
    = 
    \frac{
        p_{Z_s^e}(x_s^e - y_s^e)
    }{
        \sum_{\tilde{x}_s^e \in \mathcal{X}^\symcell} p_{Z_s^e}(\tilde{x}_s^e - y_s^e)
    }
\end{equation}
conditioned on a window \(y_s^e\) of the received vector,
for all groups \((\gstart, \gend)\), and transmit sequences \(x_{\gstart}^{\gend} \in \mathcal{X}^{\symcell}\).
This preprocessing assumes that the groups are approximately independent~\cite{duffyUsingChannelCorrelation2023a}.
Since \(Z_i\) is a Gaussian Process, the group \(Z_{\gstart}^{\gend}\) is multivariate normal distributed with covariance matrix \(C_{ij} = \sigma^2 \rho^{|i-j|}\) for \(i, j \in [\symcell]\) with 
\[ 
    p_{Z_{\gstart}^{\gend}}(z^\symcell)
    = \frac{1}{\sqrt{(2 \pi)^{\symcell} |C|}}
    \exp\left(-\frac{1}{2} (z^\symcell)^T C^{-1} z^\symcell\right).
\]
These probabilities can be processed by decoders that take group probabilities as inputs, which we refer to as \emph{group-probability decoders}. 
Examples of such decoders include ORBGRAND-AI~\cite{duffyUsingChannelCorrelation2023a} based on Symbol GRAND~\cite{anSoftDecodingSoft2022a}, non-binary LDPC~\cite{daveyLowDensityParity1998} and turbo codes~\cite{berkmannTurboDecodingNonbinary1998}.

Practical communication schemes often decode soft information for individual bits rather than groups. 
We refer to such a decoder as a \emph{bit-probability decoder} (see Fig.~\ref{fig:preprocessing-diagram}).
\begin{figure}
    \centering
    \includegraphics{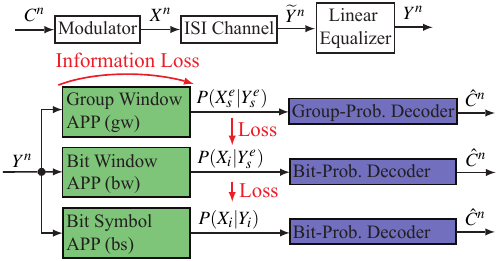}
    \caption{Preprocessing schemes: Preprocessing (gw) calculates the \ac{APP} of a group of modulation symbols \(X_s^e\) conditioned on a window of received symbols \(Y_s^e\). Preprocessing (bw) and (bs) treat individual modulation symbols as independent and can be used by a bit-probability decoder. As shown in Theorem~\ref{thm:information-rate}, the mutual information between channel input and preprocessing output decreases from processing scheme (gw) to (bw) and from (bw) to (bs).}
    \label{fig:preprocessing-diagram}
\end{figure}
The optimal preprocessing scheme for a bit-probability decoder that exploits the correlation of the same window \(y_s^e\) as preprocessing scheme (gw) in \eqref{eqn:A}, marginalizes the group probabilities of \eqref{eqn:A} to symbol probabilities%
\footnote{Pfister et al. refer to this scheme as a windowed \ac{APP}-detector~\cite{pfisterAchievableInformationRates2001}.}
\begin{equation}\label{eqn:B}
    \text{(bw)} 
    \qquad
    P_{X_i | Y_{\gstart}^{\gend}}\left(x_i | y_{\gstart}^{\gend}\right)
    =
    \sum_{t_s^e \in \mathcal{X}^\symcell: t_i = x_i}
    P_{X_{\gstart}^{\gend} | Y_{\gstart}^{\gend}}\left(
        t_{\gstart}^{\gend} | y_{\gstart}^{\gend}
    \right),
\end{equation}
for all groups \((\gstart, \gend)\) and \(x_i \in \mathcal{X}\), followed by a marginalization to bit probabilities using the gray mapping.

For comparison, we also consider a scheme that ignores any correlation between modulation symbols \(Y_i\)
\begin{equation}\label{eqn:C}
    \text{(bs)}
    \quad
    P_{X_i |Y_i}(x_i | y_i) =
    \frac{
        \exp\left((y_i - x_i)^2 / (2 \sigma^2)\right)
    }{
        \sum_{\tilde{x}_i \in \mathcal{X}}
        \exp\left((y_i - \tilde{x}_i)^2 / (2 \sigma^2)\right)
    }
\end{equation}
for all \(i \in [n]\) and \(x_i \in \mathcal{X}\).

\subsubsection{Mismatched Achievable Information Rate}
Decoding with a group- or bit-probability decoder in Fig.~\ref{fig:preprocessing-diagram} can be understood as a \emph{mismatched decoding}, where the decoder input is the product distribution of the probabilities in \eqref{eqn:A}, \eqref{eqn:B}, or \eqref{eqn:C} instead of the optimal \ac{APP} \(P(X^n | Y_{-\infty}^{+\infty})\) conditioned on the whole received sequence \(Y_{-\infty}^{+\infty}\), which reduces the capacity of the system.
An achievable information rate for a mismatched decoder is given by the \ac{GMI}~\cite{merhavInformationRatesMismatched1994, strasshoferSoftInformationPostProcessingChasePyndiah2023a}
\(I_{\tilde{s}}(X^n; Y_{-\infty}^{+\infty})\) for \(\tilde{s} \geq 0\).
For simplicity, we set \(\tilde{s}=1\) in our analysis as in~\cite{strasshoferSoftInformationPostProcessingChasePyndiah2023a}
and assume that the modulation symbols are independent and uniformly distributed.
In this case, the \ac{GMI} normalized per modulation symbol for each preprocessing becomes
\begin{align*}
    I_\text{gw} &\coloneqq 
    \frac{1}{\symcell} 
    I\left(
        X_{\gstart}^{\gend}; 
        Y_{\gstart}^{\gend}
    \right), \\
    I_\text{bw} &\coloneqq \frac{1}{\symcell} 
    \sum_{i = \gstart}^\gend
    I\left(
        X_i; Y_{\gstart}^{\gend}
    \right), 
    \quad
    I_\text{bs} \coloneqq I(X_1; Y_1),
\end{align*}
for all groups \((s, g)\) in \([n]\), where we used the fact that the process \(Y_i\) is strong stationary.
\begin{figure}
    \centering
    \includegraphics{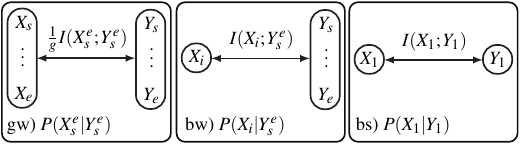}
    \caption{Calculation of the achievable information rates \(I_\text{gw}\), \(I_\text{bw}\), and \(I_\text{bs}\):
    Because \(Y_i\) is strong stationary the \ac{GMI} for \(s=1\) is equal to the mutual information between one group or symbol of the channel input and output.}
    \label{fig:capacity-analysis}
\end{figure}
The following theorem quantifies the loss in information rate between the preprocessing schemes:
\begin{theorem}\label{thm:information-rate}
    The information rates \(I_\textnormal{gw}\), \(I_\textnormal{bw}\) and \(I_\textnormal{bs}\) decreases from scheme (gw) to (bw) and from scheme (bw) to (gw), where
    \begin{align*}
        \Delta I_{\textnormal{gw-bw}} \coloneqq I_\textnormal{gw} - I_\textnormal{bw} &= \frac{1}{\symcell} \sum_{i=s+1}^e I(X_i; X_s^{i-1} | Y_s^e) \geq 0, \\
        \Delta I_{\textnormal{bw-bs}} \coloneqq I_\textnormal{bw} - I_\textnormal{bs} &= \frac{1}{\symcell} \sum_{i=s}^e I(X_i; Y_{\{s, \dots, e\} \setminus \{i\}}) \geq 0.
    \end{align*}
\end{theorem}
\begin{proof}
    Since the modulation symbols \(X_i\) are i.i.d., 
    \(I_{\textnormal{gw}} = H(X_1) - \frac{1}{\symcell} H(X_s^e | Y_s^e)\),
    \(I_\textnormal{bw} = H(X_1) - \frac{1}{\symcell} \sum_{i=s}^e H(X_i | Y_s^e)\), and
    \(I_\textnormal{bs} = H(X_1) - H(X_1 | Y_1)\).
    Taking the difference of \(I_{\textnormal{gw}}\) and \(I_\textnormal{bw}\) and applying the chain rule results in
    \[
        I_{\textnormal{gw}} - I_\textnormal{bw} =
        \frac{1}{\symcell}
        \Bigg(
            \sum_{i=s}^e \underbrace{H(X_i | Y_s^e) - H(X_i | Y_s^e, X_1^{i-1})}_{= I(X_i; X_1^{i-1} | Y_s^e)}
        \Bigg).
    \]
    Because the process \(Y_i\) and \(X_i\) are jointly strong stationary, \(H(X_1 | Y_1) = H(X_i | Y_i)\) holds for all \(i\). Therefore,
    \[
        I_\textnormal{bw} - I_\textnormal{bs} = \frac{1}{\symcell} 
        \Bigg(
            \sum_{i=s}^e \underbrace{H(X_i | Y_i) - H(X_i | Y_s^e)}_{=I(X_i; Y_{\{s, \dots, e\} \setminus \{i\}})}
        \Bigg). \qedhere
    \]
\end{proof}
Preprocessing (bs) ignores channel correlation and hence, compared to preprocessing (bw), any information that \(Y_j\) for \(j \in \{\gstart, \dots, \gend\} \setminus \{i\}\) might have about symbol \(X_i\) is lost according to Theorem~\ref{thm:information-rate}.
Although preprocessing (gw) and (bw) both take the correlation of the same groups \(Y_s^e\) into account, the calculation of \acp{APP} for individual symbols \(X_i\) in \eqref{eqn:B} instead of groups \(X_s^e\) can lead to a loss of information rate compared to (gw).
Even though the transmitted symbols are independent and hence, \(I(X_i; X_j) = 0\) for \(i \neq j\), the correlation of symbols \(Y_i\) can introduce a dependency between \(X_i\) and \(X_j\) if \(Y_s^e\) is known (i.e., \(I(X_i; X_j | Y_s^e) > 0\)).
Preprocessing (gw) can partly maintain this dependency by calculating the probability of a group of symbols.

\subsubsection{Numerical Evaluation}
To quantify the loss in information rate in Theorem~\ref{thm:information-rate}, we approximate the information rates \(I_\text{gw}\), \(I_\text{bw}\) and \(I_\text{bs}\) for the Gauss-Markov channel \eqref{eqn:channel-model-gm} via a Monte Carlo integration, which is based on~\cite[Sec.~1.5.1.3]{ryan2009channel}.
This section considers ASK modulations \(\mathcal{X}\). Since the transmission of a QAM modulation \(\mathcal{X}_{2^\numbitpermodsymb,\text{QAM}}\) can be interpreted as the transmission of two independent ASK modulations \(\mathcal{X}_{2^{\numbitpermodsymb/2},\text{ASK}}\), the information rate \(I\) of the QAM modulation is twice the information rate of the respective ASK modulation.

\begin{theorem}\label{thm:information-rate-gm}
    The information rates between channel input and decoder input of the preprocessing schemes (gw), (bw), and (bs) for the Gauss-Markov channel \eqref{eqn:channel-model-gm} are
    \begin{align*}
        I_\textnormal{gw} &= \frac{1}{\symcell} (h(Y_s^e) - h(Z_s^e)), \\
        I_\textnormal{bw} &= \frac{1}{\symcell} \sum_{i=1}^\symcell \left(h(Y_s^e) - h(Y_s^e | X_i)\right), \;
        I_\textnormal{bs} = h(Y_1) - h(Z_1),
    \end{align*}
    with
    \begin{align*}
        &h(Y_s^e) = \symcell \log(|\mathcal{X}|) - E_{Y_s^e} \Big\{
            \log\Big(
                \sum_{x^\symcell \in \mathcal{X}^\symcell} p_{Z_{\gstart}^{\gend}}(Y_s^e - x^\symcell)
            \Big)
        \Big\}, \\
        &h(Y_s^e | X_i) =
        (\symcell-1) \log(|\mathcal{X}|) \\ &- E_{Y_s^e, X_i} \Big\{
            \log\Big(\sum_{t^\symcell \in \mathcal{X}^\symcell: t_i = X_i} p_{Z_{\gstart}^{\gend}}(Y_s^e - t^\symcell)\Big)
        \Big\}, \\
        &h(Z_s^e) =
        \frac{1}{2} \log\left(
            (2 \pi e)^\symcell (1-\rho^2)^{\symcell - 1}
        \right),
    \end{align*}
    where \(C\) is a matrix with \(C_{ij} = \sigma^2 \rho^{|i-j|}\) for \(i, j \in [\symcell]\).
    \(h(Y_1)\) and \(h(Z_1)\) can be calculated via \(h(Y_s^e)\) and \(h(Z_s^e)\) with \(\symcell=1\).
\end{theorem}
\begin{proof}
    See Appendix~\ref{sec:proof-theorem}.
\end{proof}

We evaluate the expected value in Theorem \ref{thm:information-rate-gm} efficiently with a Monte Carlo integration
\(
    E_{Y_{\gstart}^{\gend}}(g(Y_{\gstart}^{\gend}))
    =
    \frac{1}{L} \sum_{i=1}^L g(y_{\gstart,i}^{\gend}),
\)
where \(y_{\gstart,i}^{\gend}\) are realizations of \(Y_{\gstart}^{\gend}\). 
\begin{figure}
    \centering
    \includegraphics{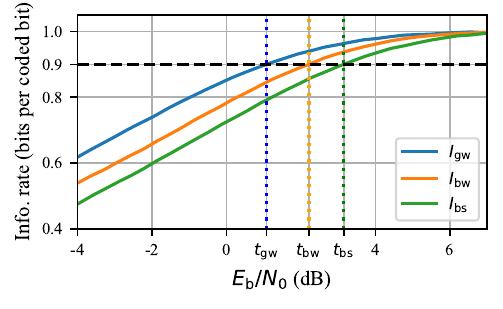}
    \caption{Achievable information rate \(I(\lEbNO)\) in bits per coded bit for each preprocessing scheme for a group size of \(\symcell=2\) modulation symbols and a Gauss-Markov channel with \(\rho=0.75\). The vertical lines indicate the \(\lEbNO\) threshold \((\lEbNO)^\star\) that is required for error-free communication of a long code of rate \(r=0.9\). The SNR gain between scheme (gw) and (bw) (\(t_\text{gw} - t_\text{bw}\)) and between (bw) and (bs) (\(t_\text{bw} - t_\text{bs}\)) can be found in Fig.~\ref{fig:snr-gain}.}
    \label{fig:inforate-curve}
\end{figure}
Figure~\ref{fig:inforate-curve} shows the information rates of Theorem~\ref{thm:information-rate-gm} over \(\lEbNO\) for a channel with \(\rho=0.75\) and group size \(\symcell=2\). As demonstrated in Theorem~\ref{thm:information-rate}, preprocessing scheme (bs) results in a loss of information compared to scheme (bw), which in turn loses information compared to scheme (gw).

This information loss translates into a loss in \(\lEbNO\) if the target \ac{BLER} is fixed: error-free communication with a long code of rate \(r\) is possible for preprocessing scheme \(\alpha \in \{\text{gw}, \text{bw}, \text{bs}\}\) if \(\lEbNO\) exceeds the \ac{SNR} threshold
\[
    (\lEbNO)^\star_\alpha \coloneqq \inf\{\lEbNO: I_\alpha(\lEbNO) \geq r\},
\]
where \(I_\alpha(\lEbNO)\) is the respective achievable information rate.
Figure~\ref{fig:snr-gain} lists the resulting gains \((\lEbNO)^\star_\text{bw} - (\lEbNO)^\star_\text{gw}\), \((\lEbNO)^\star_\text{bs} - (\lEbNO)^\star_\text{bw}\) for a code of rate \(r=0.9\).
\begin{figure*}
    \centering
    \subfloat[\(\lEbNO\) gain \((\lEbNO)^\star_\text{bw} - (\lEbNO)^\star_\text{gw}\) between scheme (gw) and (bw)]{
        \includegraphics{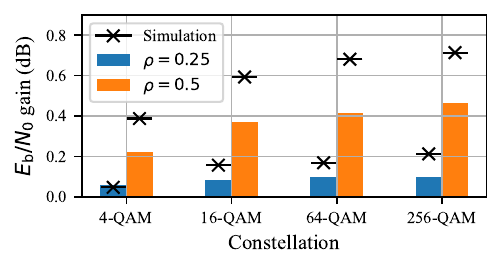}
        \label{fig:gainsAB}
    }
    \hfill
    \subfloat[\(\lEbNO\) gain \((\lEbNO)^\star_\text{bs} - (\lEbNO)^\star_\text{bw}\) between scheme (bw) and (bs)]{
        \includegraphics{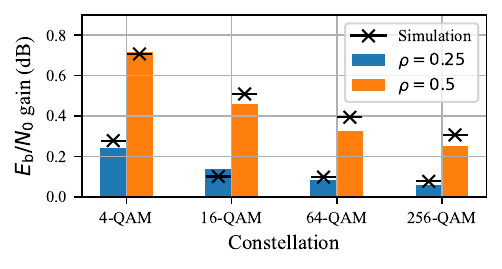}
        \label{fig:gainsBC}
    }
    \caption{\(\lEbNO\) gains between the preprocessing schemes for a group size \(\symcell=2\) symbols: The bars show the \(\lEbNO\) savings for error-free communication of a long code of rate \(r=0.9\). The black dashed lines show the actual gains for decoding an \((144, 130)\) \ac{RLCode} with ORBGRAND-AI.}
    \label{fig:snr-gain}
\end{figure*}
We compare the theoretical gains with the practical gains of ORBGRAND-AI, a group-probability decoder for small component codes~\cite{duffyUsingChannelCorrelation2023a}.
The black lines in Fig.~\ref{fig:inforate-curve} show the actual gains that ORBGRAND-AI achieves for the respective preprocessing scheme for a \((144, 130)\) \acf{RLCode}, i.e., a \((144, 130)\) linear code whose binary parity-check entries are sampled independently and uniformly at random.
The decoding gains closely follow the theoretical gains between preprocessing (bw) and (bs). 
While ORBGRAND-AI achieves higher gains than predicted between (gw) and (bw), the analysis correctly predicts the gain increase with constellation order.

\subsection{Endogenous Correlation}\label{sec:theo-endogenous-corr}
The previous section demonstrates that group probabilities mitigate information loss for exogenous bit correlation. Here, we show that the decoder introduces bit correlation on its own, which we refer to as \emph{endogenous correlation}.
If a component decoder is used for iterative decoding, message passing of group instead of bit probabilities, can mitigate the information loss caused by marginalization and improve decoding performance.

To analyze this effect, we consider a \ac{BI-AWGN} channel, i.e., \(\rho=0\) and \(\mathcal{X} = \{-1, +1\}\) in \eqref{eqn:channel-model-gm} that does not introduce any exogenous correlation between the codeword bits \(C^n\) (see Fig.~\ref{fig:siso-decoder}).
\begin{figure}
    \centering
    \includegraphics{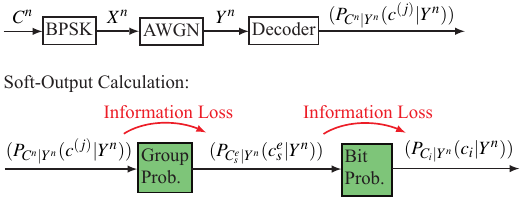}
    \caption{Information processing chain: information loss under group and bit marginalization:
    The decoder outputs a list of codewords and corresponding \acp{APP}. The codeword \acp{APP} are marginalized to group probabilities for consecutive groups of size \(\symcell\), followed by a marginalization to bit probabilities. According to the chain and the data-processing inequality~\cite{SWB-1644451026}, the mutual information between the decoder output and the probabilities can only remain constant or decrease from group to bit probabilities.}
    \label{fig:siso-decoder}
\end{figure}
A \ac{SISO} component decoder that outputs optimal block-wise \ac{SO} information calculates the \ac{APP} of each codeword \(c^n \in \mathcal{C}\) as follows:
\begin{equation} \label{eqn:decoder-optimal}
    P_{C^n|Y^n}(c^n | y^n)
    =
    \frac{
        P_{X^n | Y^n}(\text{BPSK}(c^n) | y^n)
    }{
        \sum_{t^n \in \mathcal{C}} P_{X^n | Y^n}(\text{BPSK}(t^n) | y^n)
    },
\end{equation}
where \(\text{BPSK}(\cdot)\) denotes the mapping from a binary to BPSK sequence and \(P_{X^n | Y^n}(x^n | y^n)\) is the statistic of the BI-AWGN channel.
Such a decoder can be approximated via list decoding, such as the Chase-Pyndiah algorithm~\cite{pyndiahNearoptimumDecodingProduct1998}, and recent improvements like SOGRAND~\cite{yuanSoftoutputGRANDLong2023}, SO-SCL~\cite{yuanSoftOutputSuccessiveCancellation2025}, SO-GCD~\cite{duffySoftOutputGuessingCodeword2025}, and SOCS decoding~\cite{Janz25soft}, which also provide a probability that the correct codeword is not in the list.

When the decoder is used as a component decoder for turbo product decoding with bit probabilities as in \cite{pyndiahNearoptimumDecodingProduct1998, duffySoftOutputGuessingCodeword2025, yuanSoftoutputGRANDLong2023, Janz25soft}, bit probabilities need to be extracted from the codeword list to be passed to the next row or column decoder.
A posteriori bit probabilities can be extracted from \eqref{eqn:decoder-optimal} via marginalization
\begin{equation}\label{eqn:decoder-marg-group}
    P_{C_i | Y^n}(c_i | y^n) 
    = \sum_{\mathclap{t^n \in \mathcal{C}: t_i = c_i}} P_{C^n|Y^n}(t^n | y^n).
\end{equation}

Motivated by the previous section, we consider iterative decoding with group probabilities, which marginalizes the output to group probabilities before passing them to the next component decoders:
\begin{equation}\label{eqn:decoder-marg-bit}
    P_{C_s^e | Y^n}(c_s^e| y^n) = \sum_{\mathclap{t^n \in \mathcal{C}: t_s^e = c_s^e}} P_{C^n|Y^n}(t^n | y^n).
\end{equation}
To see to what degree the group probabilities improve iterative decoding, we compare the mutual information per coded bit after the marginalizing step for group probabilities \(I_\text{g} = \frac{1}{\symcell} I(C_s^e; Y^n)\) with the mutual information for bit probabilities \(I_\text{b} = \frac{1}{\symcell} \sum_{i = s}^e I(C_i; Y^n)\) resulting in
\begin{equation}\label{eqn:info-loss-decoder}
    I_\text{g} - I_\text{b} = \frac{1}{\symcell} \sum_{i = s+1}^e I(C_i ; C_s^{i-1} | Y^n) \geq 0.
\end{equation}

To quantify the information loss in \eqref{eqn:info-loss-decoder}, we use a Monte Carlo simulation, which samples received vectors \(y^{(j), n}\) (\(j \in [L]\)) for \(L\) uniformly distributed codewords transmitted over a BI-AWGN channel as well as a \((n, k)\) \ac{RLCode} in every trial. \(I_\text{g}\) and \(I_\text{b}\) can then be approximated via
\begin{align*}
    I_\text{g}
    &= 1 - \frac{1}{\symcell} E_{Y^n} \{ H(C_s^e | Y^n = Y^n) \} \\
    &\approx
    1 - \frac{1}{\symcell L}
    \sum_{j = 1}^L H(C_s^e | Y^n = y^{(j),n})
    =
    1 + \frac{1}{\symcell L}
    \sum_{j = 1}^L \sum_{c^\symcell \in \{0, 1\}^\symcell} \\
    &\hspace{1.5em} P_{C_s^e | Y^n}(c^\symcell | y^{(j), n})
    \log(P_{C_s^e | Y^n}(c^\symcell | y^{(j), n})),\\
    I_\text{b} 
    &\approx 1 - \frac{1}{L} \sum_{j = 1}^L H_\text{b}(P_{C_i | Y^n}(0 | y^{(j), n})),
\end{align*}
where the probabilities are calculated for each received vector \(y^{(j), n}\) using \eqref{eqn:decoder-optimal}, \eqref{eqn:decoder-marg-group}, \eqref{eqn:decoder-marg-bit}, and the binary entropy function \(H_\text{b}(\cdot)\).

\begin{figure}
    \centering
    \includegraphics{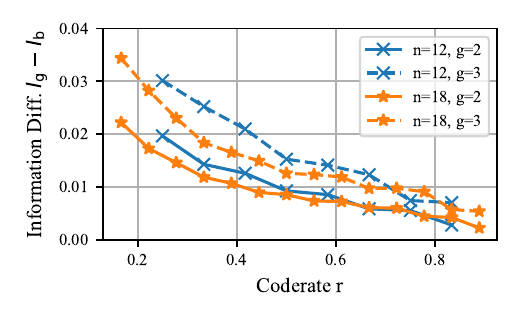}
    \caption{Difference in mutual information between codeword groups and channel output \(I_\text{g} = \frac{1}{\symcell} I(C_s^e ; Y^n)\), and between codeword bits and channel output \(I_\text{b} = \frac{1}{\symcell} \sum_{i = s}^e I(C_i ; Y^n)\): The information difference is evaluated for different \acp{RLCode} with different code parameters \((n, k)\).
    }
    \label{fig:info-difference-decoding}
\end{figure}

Fig.~\ref{fig:info-difference-decoding} shows numerical results of the gap between \(I_\text{g}\) and \(I_\text{b}\) for \(\lEbNO = \SI{3}{\dB}\)
evaluated for \acp{RLCode} of different rates and lengths.%
\footnote{We choose \(\lEbNO = \SI{3}{\dB}\) because the non-binary \acp{TPC}, used later to analyze the benefits of preserving endogenous correlation, achieve a target BLER between \num{e-2} and \num{e-3} at this \(\lEbNO\).}
As expected, the difference in mutual information increases with increasing group size \(\symcell\).
Notably, the difference increases as the rate of the code decreases, indicating that product codes with low-rate component codes benefit more from group-probability decoding than high-rate codes. In Sec.~\ref{sec:sim-endogenous-corr}, we provide empirical support for this observation by decoding product codes with component codes of varying rates.

This finding can be motivated by examining two boundary cases. Consider the binary code of length \(n\) of rate \(1\), which contains all binary vectors of length \(n\). Since the channel is memoryless, the codeword probabilities are the product distribution of bit probabilities, and marginalization to bit probabilities does not introduce any information loss.
In contrast, for a low-rate repetition code containing only the all-zero and all-one codeword, the codeword probabilities cannot, in general, be represented by a product distribution.
In this case, group probabilities are only positive for groups of identical bits, partially capturing the bit correlation lost with bit probabilities.
As the code becomes sparser, the discrepancy between codeword probabilities and product distributions grows, and the advantage of group probabilities becomes more significant.

\section{Non-binary Product Codes}\label{sec:non-binary-tpc}
A binary product code is a binary matrix, where a component code protects every row and column. These codes can be efficiently decoded with turbo product decoding~\cite{pyndiahNearoptimumDecodingProduct1998}, where \acp{LLR}, i.e., bit probabilities, are iteratively updated. To update group probabilities instead, the bits of each group must be part of the same row and column codeword. This can be achieved by placing bit groups, rather than individual bits, into the product code cells, as demonstrated in Fig.~\ref{fig:code-construction}. This section first presents the construction of these non-binary product codes, followed by the turbo product decoding algorithms.
We compare two cell reliability measures for decoding: \emph{bit probabilities} and \emph{group probabilities}.

\paragraph{Bit Probabilities}
Every bit in each cell is assigned \acp{LLR}, which are updated in several decoding iterations until a termination condition is fulfilled. This approach can be used if the correlation between bits is ignored, as in preprocessing (bw) and (bs) in \eqref{eqn:B} and \eqref{eqn:C}.

\paragraph{Group Probabilities}
Each product code cell is assigned group probabilities. Unlike the bit probability update, this approach captures both exogenous correlation from preprocessing (gw) in \eqref{eqn:A} and endogenous correlation from the component decoder.

\subsection{Code Construction}
\begin{figure}
    \centering
    \includegraphics{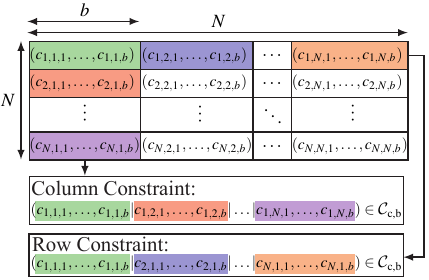}
    \caption{Non-binary product code with \(N\) rows and columns, and \(\bitcell\) bits per cell: For every row and column, the concatenation of the bits in each cell needs to form a valid codeword of the component code \(\cc_\text{c,\bitcell}\).}
    \label{fig:code-construction}
\end{figure}

Let \(b\) be the number of bits per product code cell corresponding to \(g = b / \ell\) consecutive modulation symbols per cell. Let \(\text{GF}(2^\bitcell)\) be an extension field over \(\text{GF}(2)\) and let the component code \(\ccc\) be a systematic \((N, K)\) linear block code  over \(\text{GF}(2^\bitcell)\).
Non-binary product codes \cite[Sec.~10.4]{moonErrorCorrectionCoding2005}
\(\ccc \times \ccc\) are \(N \times N\) matrices over \(\text{GF}(2^\bitcell)\) for which each row and each column is a codeword of a non-binary linear code \(\ccc\).
To encode these codes, first \(K^2\) elements of \(\text{GF}(2^\bitcell)\) are ordered in a \(K \times K\) matrix. Then, each row is encoded with the systematic code \(\ccc\), adding \(N-K\) parity \(\text{GF}(2^\bitcell)\)-elements per row. Next, each column is encoded with \(\ccc\), resulting in a \(N \times N\) matrix.  Owing to linearity, the last \(N-K\) rows are also codewords of \(\ccc\) \cite[Sec.~10.4]{moonErrorCorrectionCoding2005}. This construction results in a code of rate \(R = \left(K / N\right)^2\).

We reformulate the above definition in the binary domain.
Although some soft-input component list decoders like \ac{GRAND} and \ac{GCD} can directly decode the non-binary component code \(\ccc\), expressing the code in binary removes the need for Galois-field arithmetic, enabling the reuse of existing binary decoder hardware.
To transform the code, we represent the elements of \(\text{GF}(2^\bitcell)\) as binary vectors of length \(\bitcell\) using a basis of \(\text{GF}(2^\bitcell)\) over \(\text{GF}(2)\) \cite{macwilliamsTheoryErrorCorrecting1978}.
By mapping each symbol of a codeword of \(\ccc\) to its binary representation and concatenating them, we obtain a binary vector of length \(\bitcell N\). The set of the binary representations of all codewords forms an \((\bitcell N, \bitcell K)\) binary linear code \(\cccb\), known as the \emph{binary image} of \(\ccc\) \cite[Ch.~10. §5]{macwilliamsTheoryErrorCorrecting1978}.

The binary representations of the non-binary product code \(\ccc \times \ccc\) are \(N \times N\) matrices \(c^{N,N,\bitcell} \in \{0, 1\}^{N \times N \times \bitcell}\) where each entry \(c_{i, i, :} \in \{0, 1\}^\bitcell\) contains \(\bitcell\) bits.
A matrix \(c^{N, N,\bitcell}\) is a valid codeword of the product code if, for every row and column, the concatenated binary entries form a codeword of \(\cccb\):
\begin{align}
    &(c_{i, 1, 1}, \dots, c_{i, 1, \bitcell} \mid
    \dots \mid c_{i, N, 1}, \dots, c_{i, N, \bitcell}) \in \cccb,
    \label{eqn:code-construction-1} \\
    &(c_{1, j, 1}, \dots, c_{1, j, \bitcell} \mid
    \dots \mid c_{N, j, 1}, \dots, c_{N, j, \bitcell}) \in \cccb,
    \label{eqn:code-construction-2}
\end{align}
for each row \(i\) and column \(j \in \{1, \dots, N\}\), where \(c_{i, j, s}\) is the \(s\)-th bit of the product code cell \(c_{i, j}\).

Fig.~\ref{fig:code-construction} illustrates the code construction.
To encode a \(u^{K, K,\bitcell} \in \{0, 1\}^{K \times K \times \bitcell}\) matrix of information bits, the cells of each row are concatenated according to \eqref{eqn:code-construction-1}, encoded, and written back into the cells. This process is then repeated for all columns. The result is a valid codeword \(c^{N,N,\bitcell}\) that fulfills both constraints \eqref{eqn:code-construction-1} and \eqref{eqn:code-construction-2}. Encoding works because the underlying binary component code \(\cccb\) is the binary image of a linear code over \(\text{GF}(2^{\bitcell})\).
In contrast, if an arbitrary binary component code is used, the last \(N-K\) rows of the encoded matrix \(c^{N, N,\bitcell}\) are unlikely to be component codewords, which would result in an invalid product code codeword as demonstrated in Fig.~\ref{fig:counterexample-encoding}.

\begin{figure}
    \centering
    \begin{minipage}{0.45\linewidth}
        \includegraphics{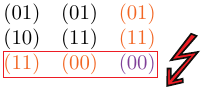}
    \end{minipage}
    \begin{minipage}{0.48\linewidth}
        \begin{equation*}
            G = 
            \begin{pmatrix}
                1 &  &  &  & 0 & 1\\
                 & 1 &  &  & 1 & 0\\
                 &  & 1 &  & 0 & 1\\
                 &  &  & 1 & 1 & 1
            \end{pmatrix}
        \end{equation*}
    \end{minipage}
    \caption{Encoding of a non-binary product code resulting in invalid codeword: The systematic generator matrix of the binary component code is shown on the right, and the encoded modified product code is shown on the left. Parity check bits are colored orange and purple (checks on checks). The last row is not a valid codeword because the checks on checks were encoded via column encodings.}
    \label{fig:counterexample-encoding}
\end{figure}

For our numerical results, we use the binary image of non-binary \ac{RS} and extended \ac{RS} codes: an \((N, K)\) \ac{RS} code is defined over an extension field \(\text{GF}(2 ^ \bitcell)\) with code
length \(N = 2 ^ \bitcell - 1\) and arbitrary information length \(K < N\). Singly extended RS codes add one parity check symbol
\(c_{N+1} = -\sum_{i=1}^N c_i\) at the end of each codeword \((c_1, \dots , c_N)\)
of an \((N,K)\) RS code resulting in an \((N +1,K)\) linear block
code.

\subsection{Turbo Product Decoding with Bit Probabilities}
Bit-probability turbo product decoding of non-binary \acp{TPC} conceptually follows Pyndiah's algorithm~\cite{pyndiahNearoptimumDecodingProduct1998}, adapted to account for the non-binary product code structure. The decoding process is outlined below for an arbitrary bit-probability \ac{SISO} decoder. Each bit \(c_{i, j, s}\) is assigned a channel, a priori and a-posteriori \ac{LLR} \(L_{\text{Ch}, i, j, s}\), \(L_{\text{A}, i, j, s}\) and \(L_{\text{APP}, i, j, s}\), respectively. 
Channel LLRs are obtained from equation \eqref{eqn:B} or \eqref{eqn:C} by marginalizing the modulation symbol probabilities to bit probabilities.
Initially, the a priori \acp{LLR} are set to \(L_{\text{A}, i, j, s} = 0\), assuming that the codeword bits are uniformly distributed at random.
The a priori \acp{LLR} are iteratively updated during \(I_\text{max}\) decoding iterations. Each decoding iteration consists of decoding all columns followed by decoding all rows. In the following, we refer to one column or row decoding as a \emph{half-iteration}. During one half-iteration, the following operations are performed:
\begin{itemize}[leftmargin=*]
    \item The channel and a priori \acp{LLR} of each \(i\)-th row or column are concatenated according to \eqref{eqn:code-construction-1} or \eqref{eqn:code-construction-2}, respectively, resulting in the vectors
    \(
        L^{\bitcell N,(i)}_{\text{Ch}}, L^{\bitcell N,(i)}_{\text{A}} \in \rb^{\bitcell N}
    \). For instance, to decode the \(i\)-th row, \(L^{\bitcell N,(i)}_{\text{Ch}}\) is
    \[
        (L_{\text{Ch}, i, 1, 1}, \dots, L_{\text{Ch}, i, 1, \bitcell}, \dots, L_{\text{Ch}, i, N, 1}, \dots, L_{\text{Ch}, i, N, \bitcell}).
    \]
    \item The sum of both vectors \(L^{\bitcell N,(i)}_{\text{Ch}} + L^{\bitcell N,(i)}_{\text{A}}\) is input into a bit-probability \ac{SISO} decoder, which outputs a vector of a posteriori \acp{LLR} \(L^{\bitcell N,(i)}_{\text{APP}}\). 
    \item The extrinsic information of the current row or column is calculated as
    \(
        L^{\bitcell N,(i)}_{\text{E}} = L^{\bitcell N,(i)}_{\text{APP}} - L^{\bitcell N,(i)}_{\text{Ch}} - L^{\bitcell N,(i)}_{\text{A}}
    \),
    which is used to update the priori \ac{LLR} vector
    \(L^{\bitcell N,(i)}_{\text{A}} \leftarrow \alpha L^{\bitcell N,(i)}_{\text{E}}\) using a dampening factor of \(\alpha \in [0, 1]\).
    \item The \acp{LLR} of the vectors \(L^{\bitcell N,(i)}_{\text{APP}}\) and \(L^{\bitcell N,(i)}_{\text{A}}\) are written back into the respective cells according to \eqref{eqn:code-construction-1} or \eqref{eqn:code-construction-2}.
\end{itemize}
At the end of each half iteration, the hard decision output is calculated as
\(
    w_{\text{hd}, i, j, s} = \ind_{\{L_{\text{APP}, i, j, s} < 0\}}.
\)
If \(I_\text{max}\) full iterations have passed or \(w_\text{hd}^{N,N,\bitcell}\) is a valid codeword, i.e., fulfills \eqref{eqn:code-construction-1} and \eqref{eqn:code-construction-2}, the decoding is terminated and \(w_\text{hd}^{N,N,\bitcell}\) is output.

Existing studies of bit-probability decoding of non-binary \acp{TPC} typically use Pyndiah's algorithm~\cite{pyndiahNearoptimumDecodingProduct1998} as \ac{SISO} decoder. In contrast, our analysis is based on \ac{SOGRAND} with 1-line ORBGRAND~\cite{duffyOrderedReliabilityBits2022, galliganBlockTurboDecoding2023} as it provides more accurate soft information and can handle any linear component code~\cite{yuanSoftoutputGRANDLong2023}.

\subsection{Turbo Product Decoding with Group Probabilities}\label{sec:turbo-group-probs}

\newcommand{\numgrouplist}{n_{\mathcal{B}}}

\begin{figure}
    \centering
    \includegraphics{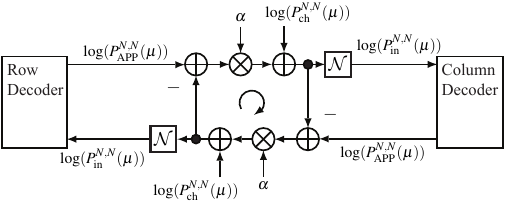}
    \caption{Turbo decoding of non-binary product codes}
    \label{fig:turbo-decoding-non-binary}
\end{figure}

Analogously to the previous section, we assign each cell \((i,j)\) group probabilities \(P_{\text{Ch}, i, j}(\mu)\), \(P_{\text{A}, i, j}(\mu)\) and \(P_{\text{APP}, i, j}(\mu)\), which are the channel, a priori and a-posteriori probability
of the group \(c_{i,j, :}\) in cell \((i,j)\) being equal to \(\mu \in \{0, 1\}^\bitcell\). The channel probabilities are obtained  from \eqref{eqn:A}:
\[
    P_{\text{Ch}, i, j}(\mu) \coloneqq P_{C_{i,j,:} | Y_{i,j,:}}(\mu | y_{i,j,:}), \quad \text{for \(\mu \in \{0, 1\}^\bitcell\),}
\]
where \(Y_{i,j,:}\) denotes the received \(g = b / \ell\) modulation symbols corresponding to product code cell \((i,j)\).
To reduce the storage and decoding complexity, only the top \(\numgrouplist \in \fnatural\) bit groups \(\mathcal{B}_{i,j} \subset \{0, 1\}^\bitcell\) with the highest channel probabilities \(P_{\text{Ch}, i, j}(\mu)\) are retained for each cell \((i, j)\)~\cite{anSoftDecodingSoft2022a}. The probabilities of the remaining symbols are set to \(0\) and \(P_{\text{Ch}, i, j}(\mu)\) is renormalized.
The a priori probabilities are initialized with a uniform distribution \(P_{\text{A}, i, j}(\mu) = 1 / \numgrouplist\) for \(\mu \in \mathcal{B}_{i,j}\) and \(0\) otherwise.
Analogously to \acp{LLR}, we represent the group probabilities as log probabilities for numerical stability. To ensure that the probabilities represented by the log probabilities sum up to \(1\), we normalize them before inputting them into the \ac{SISO} decoder. That means, for a vector of log probabilities \(v^n\), the normalization is
\[
    \mathcal{N}(v^n)_i \coloneqq v_i - \log\left(\sum_{i=1}^n \exp(v_i)\right), \quad \text{for \(i \in [n]\)},
\]
which can be efficiently calculated with the Jacobian Logarithm~\cite{erfanianReducedComplexitySymbol1994a}.

During each half iteration, columns and rows are alternately decoded. When rows are decoded, let
\(P_{\text{Ch}}^{N,(i)}\), \(P_{\text{A}}^{N,(i)}\), \(P_{\text{APP}}^{N,(i)}\) denote the \(i\)-th row of \(P_{\text{Ch}}^{N,N}\), \(P_{\text{A}}^{N,N}\), and \(P_{\text{APP}}^{N,N}\), respectively.
When columns are decoded, they denote the \(i\)-th column.
During a half iteration, the following operations are performed for the \(i\)-th row or column (see Fig.~\ref{fig:turbo-decoding-non-binary}):
\begin{itemize}
    \item Calculate the component decoder input 
    \begin{equation*}
        \log(\tilde{P}_{\text{in}}^{N,(i)}(\mu)) \coloneqq \log(P_{\text{Ch}}^{N,(i)}(\mu)) + \log(P_{\text{A}}^{N,(i)}(\mu)).
    \end{equation*}
    \item Normalize the log probabilities 
    \begin{equation}\label{eqn:symbol-input}
        \log(P_{\text{in},j}^{N,(i)}(\mu)) = \mathcal{N}(\log(\tilde{P}_{\text{in},j}^{N,(i)}(\mu)))
        \; \text{for \(j \in [N]\).}
    \end{equation}
    \item Decode the probability vector \(\log(P_{\text{in}}^{N,(i)}(\mu))\) with a group-probability \ac{SISO} decoder (Sec.~\ref{sec:group-prob-siso-decoder}) that returns the \ac{SO} vector \(\log(P_{\text{App},i}^N(\mu))\).
    \item Calculate the extrinsic probabilities
    \begin{multline*}
        \log(P_{\text{E}}^{N,(i)}(\mu))
        =
        \log(P_{\text{APP}}^{N,(i)}(\mu)) \\
        - \log(P_{\text{Ch}}^{N,(i)}(\mu)) 
        - \log(P_{\text{A}}^{N,(i)}(\mu))
    \end{multline*}
    and update
    \(
        \log(P_{\text{A}}^{N,(i)}(\mu)) \leftarrow \alpha \log(P_{\text{E}}^{N,(i)}(\mu))
    \) for the next half-iteration.
\end{itemize}
At the end of a half iteration, the hard decision output is calculated as
\[
    y_{\text{hd}, i, j} = \underset{\mu \in \{0, 1\}^\bitcell}{\argmax} \log(P_{\text{APP},i,j}(\mu)).
\]
Decoding terminates with \(y_{\text{hd}}^{N,N}\) as output when either \(I_\text{max}\) full iteration have been completed or \(y_{\text{hd}}^{N,N}\) is a valid codeword.

\subsection{Group-Probability \ac{SISO} Decoder}\label{sec:group-prob-siso-decoder}
Group-probability \ac{SISO} component decoding can be performed with any list decoder that takes group probabilities \(P_\text{in}^N(\mu)\) as input and outputs a list of potential component codewords \(\mathcal{L} \subset \mathcal{C}\) and block-wise \ac{SO}. The block-wise \ac{SO} consists of an \ac{APP} 
\(
    P_{C^N | P_\text{in}^N}\left(
        c^N \mid P_\text{in}^N
    \right)
\)
for each codeword \(c^N \in \mathcal{L}\) and a probability that the correct codeword is not in the list \(P_{C^N | P_\text{in}^N}(\mathcal{C} \setminus \mathcal{L}\mid P_\text{in}^N)\).
After list decoding, the \acp{APP} \(P^N_\text{APP}(\mu)\) are calculated for each group position \(i \in [N]\) and bit group \(\mu \in \mathcal{B}_i\), where \(\mathcal{B}_i \subset \{0, 1\}^\bitcell\) denotes the \(\numgrouplist\) bit groups considered for at position \(i \in [N]\) of the sequence (see Sec.~\ref{sec:turbo-group-probs}). 
The \acp{APP} \(P^N_\text{APP}(\mu)\) are a weighted sum of prior beliefs with the likelihood that the correct codeword is not contained in the list, with a sum of the codeword likelihoods in the list~\cite{yuanSoftoutputGRANDLong2023}:
\begin{multline*}
    P_{\text{APP},i}(\mu)
    =
    \sum_{c^n \in \mathcal{L}: c_i = \mu} 
    P_{C^N | P_\text{in}^N}\left(
        c^N \mid P_\text{in}^N
    \right) \\
    +
    P_{C^N | P_\text{in}^N}(\mathcal{C} \setminus \mathcal{L}\mid P_\text{in}^N)
    P_{\text{in},i}(\mu),
\end{multline*}
Yuan et al.~\cite{yuanSoftoutputGRANDLong2023} demonstrated for binary \acp{TPC} that the second weighting term improves the accuracy of the bit-wise \ac{SO}. However, even if the component decoder does not provide \(P_{C^N | P_\text{in}^N}(\mathcal{C} \setminus \mathcal{L}\mid P_\text{in}^N)\), e.g., in Pyndiah's original approach~\cite{pyndiahNearoptimumDecodingProduct1998}, \ac{SO} can still be calculated via classical marginalization by implicitly setting \(P_{C^N | P_\text{in}^N}(\mathcal{C} \setminus \mathcal{L}\mid P_\text{in}^N)\) to \(0\).

Recently, practical decoders have been developed that process group probabilities, such as symbol-level ORBGRAND~\cite{anSoftDecodingSoft2022a} combined with \ac{SOGRAND}~\cite{yuanSoftoutputGRANDLong2023}, or \ac{GCD} combined with the pattern generator of \cite{anSoftDecodingSoft2022a} and SO-GCD~\cite{duffySoftOutputGuessingCodeword2025}.

\subsubsection{Group-Probability \ac{SOGRAND} Decoding}\label{sec:group-prob-sogrand}
For the numerical results in this paper, we use a variant of \ac{SOGRAND} as the group probability \ac{SISO} decoder, which is outlined in the following.
Our decoder combines concepts of several \ac{GRAND}~\cite{duffyCapacityAchievingGuessingRandom2019} variants: symbol-level ORBGRAND~\cite{anSoftDecodingSoft2022a}, list decoding~\cite{abbasListGRANDPracticalWay2023}, and \ac{SOGRAND}~\cite{yuanSoftoutputGRANDLong2023}.
\ac{GRAND} decodes by subtracting noise patterns in decreasing order of probability from the received hard-decision sequence until a valid codeword is found. This principle allows \ac{GRAND} to decode any component code of moderate redundancy, including non-linear codes~\cite{cohenAESErrorCorrection2023}.  Variants have been developed for different channel models, including the binary-symmetric channel~\cite{duffyCapacityAchievingGuessingRandom2019}, soft-decision decoding~\cite{solomonSoftMaximumLikelihood2020b, duffyOrderedReliabilityBits2022, anSoftDecodingSoft2022a}, and correlated channel~\cite{duffyUsingChannelCorrelation2023a, anKeepBurstsDitch2022}. Practical implementations have been demonstrated through hardware syntheses~\cite{condoHighperformanceLowcomplexityError2021, condoFixedLatencyORBGRAND2022, abbasHighThroughputEnergyEfficientVLSI2022} and taped out chips~\cite{riazSub08pJ163GbpsMm22023, riazSub08pJ163GbpsMm22023}.

First, the hard-decision \(w_{\text{hd}}^N \in \mathcal{B}_1 \times \dots \times \mathcal{B}_N\) is calculated for each group \(i \in [N]\) as
\[
    w_{\text{hd},i} = \argmax_{\mu \in \mathcal{B}_i} \log(P_{\text{in},i}(\mu)).
\]
Next, for each group \(i \in [N]\), and for all symbols \(\mu \in \mathcal{B}_i \setminus \{w_{\text{hd},i}\}\), the \ac{LLR} 
\[
    \delta_i(\mu) \coloneqq 
    \log(P_{\text{in},i}(w_{\text{hd},i})) - \log(P_{\text{in},i}(\mu))
\]
is calculated using the input probabilities defined in \eqref{eqn:symbol-input}.%
\footnote{In \cite{anSoftDecodingSoft2022a}, \(\delta_i(\mu)\) are referred to as exceedance distances measuring the distance between a modulation symbol and the hard-decision modulation symbol. For an AWGN channel, these distances can be interpreted as \acp{LLR} as used in this paper.}
Symbol-level ORBGRAND takes these \acp{LLR} \(\delta_i(\mu)\) and hard-decision \(w_{\text{hd}}^N\) as input and iterates efficiently over all possible sequences \(w^{N,(j)} \in \mathcal{B}_1 \times \dots \times \mathcal{B}_N\) in decreasing order of their probability
\[
    P_{Y^N | P_\text{in}^N}(w^{(j),N} | P_\text{in}^N)
    \coloneqq
    \prod_{i =1}^N P_{\text{in},i}(w_{i}^{(j)})
\]
(see~\cite{anSoftDecodingSoft2022a} for details), where \(j\) is the index of the guessing order. We use 1-line ORBGRAND as the pattern generator~\cite{duffyOrderedReliabilityBits2022, galliganBlockTurboDecoding2023} in symbol-level ORBGRAND, as it offers a more accurate approximation of the \acp{LLR} statistics during turbo product decoding than basic ORBGRAND~\cite{galliganBlockTurboDecoding2023}.

Let \(H^{n-k,n} \in \{0, 1\}^{(n-k) \times n}\) be the binary parity check matrix of the component code \(\cccb\).
For each sequence \(w^{N,(j)}\), the following steps are performed:
\begin{itemize}[leftmargin=*]
    \item Concatenate the groups of the guess \(w^{N,(j)} \in (\{0, 1\}^b)^N\) to a binary vector \(w^{N \bitcell, (j)}\). If \(w^{N \bitcell, (j)}\) is a codeword, i.e., 
    \[
        H^{n-k,n} w^{N \bitcell, (j)} = 0^{n-k},
    \] 
    add \(w^{N, (j)}\) to the list \(\mathcal{L}\).
    \item An estimate of the probability that the correct codeword is not in the list is calculated~\cite{yuanSoftoutputGRANDLong2023}
    \[
        P_{C^N \mid P_\text{in}^N}(\mathcal{C} \setminus \mathcal{L} \mid P_\text{in}^N) 
        \coloneqq
        \frac{P(A)}{
            \sum\limits_{c^N \in \mathcal{L}} 
            P_{Y^N | P_\text{in}^N}\left(
                c^N \mid P_\text{in}^N
            \right) + P(A)
        },
    \]
    where
    \[
        P(A) \coloneqq \Bigg(1 - \sum_{u=1}^{j} 
            P_{Y^N | P_\text{in}^N}\left(
                w^{N,(u)} | P_\text{in}^N
            \right)
        \Bigg) \frac{2^{bK} - 1}{2^{bN} - 1}
    \]
    incorporates the noise probabilities from guess \(1\) to \(j\).
    \item If \(\mathcal{L}\) contains \(n_\mathcal{L}\) codewords or \(P_{C^n \mid P_\text{in}^N}(\mathcal{C} \setminus \mathcal{L} \mid P_\text{in}^N) \) is lower than a threshold \(T \in [0, 1]\), list decoding terminats early because the correct codeword is in the list with probability greater than \(1-T\).
\end{itemize}
After list decoding, the codeword \acp{APP} of the block-wise \ac{SO} are calculated as~\cite{yuanSoftoutputGRANDLong2023}
\[
    P_{C^N | P_\text{in}^N}\left(
        c^N \mid P_\text{in}^N
    \right)
    \coloneqq
    \frac{
        P_{Y^N \mid P_\text{in}^N}(
            c^N \mid P_\text{in}^N
        )
    }{
        \sum_{\widetilde{c}^N \in \mathcal{L}} 
        P_{Y^N \mid P_\text{in}^N}(
            \widetilde{c}^N \mid P_\text{in}^N
        ) + P(A)
    }.
\]

\section{Results}\label{sec:res-group-prob}
This section presents simulation results for the decoding performance of non-binary \acp{TPC} for endogenous and exogenous correlation.
The maximal list size is \(n_\mathcal{L} = 4\), \(T = \num{e-4}\) and \(\alpha=0.5\).

\subsection{Endogenous Correlation}\label{sec:sim-endogenous-corr}
The analysis in Sec.~\ref{sec:theo-endogenous-corr} demonstrates that using bit instead of group probabilities results in an information loss during the first decoding iteration that increases with decreasing rate of the component code. To demonstrate how this loss affects the decoding, we constructed non-binary \acp{TPC} based on the \ac{RS} \((7, 3)\), \((7,4)\), and \((7, 5)\) component codes. All component codes have the same length but differ in rate and result in a non-binary \ac{TPC} with \(\bitcell=3\) bits per cell. As in Sec.~\ref{sec:theo-endogenous-corr}, the codes are transmitted over a \ac{BI-AWGN} channel, i.e., the received bits are uncorrelated. 

Figure~\ref{fig:endogenous-result} shows the \ac{BLER} performance for different numbers of half iterations. Each curve corresponds to a different number of half iterations after which decoding terminates. The simulation results show that decoding with group probabilities improves the decoding performance compared to bit-probability decoding if the rate of the component code is sufficiently low. Consistent with our analysis, the gains increase as the rate of the component code decreases.
The component decoder in the first half iteration introduces endogenous correlation between the bits, which is captured by the group probabilities and facilitates the component decoding in the second iteration, thereby improving the \ac{BLER} performance.
After convergence, group-probability decoding achieves gains of \SI{0.29}{\dB} and \SI{0.21}{\dB} over bit-probability decoding for the RS (7, 3) \ac{TPC} and RS (7, 4) \ac{TPC}, respectively, at a \ac{BLER} of \num{e-4}.
\begin{figure*}
    \centering
    \subfloat[RS (7, 3) component code.]{
        \includegraphics[height=0.23\linewidth]{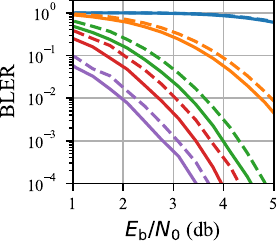}
        \label{fig:rs73}
    }
    \subfloat[RS (7, 4) component code.]{
        \includegraphics[height=0.23\linewidth]{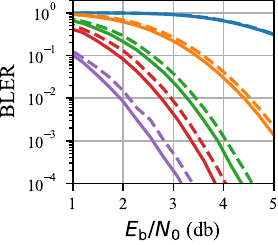}
        \label{fig:rs74}
    }
    \subfloat[RS (7, 5) component code.]{
        \includegraphics[height=0.23\linewidth]{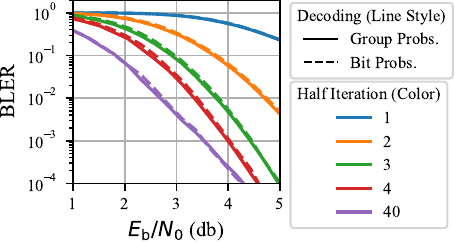}
        \label{fig:rs75}
    }
    \caption{Turbo decoding performance of different non-binary \ac{TPC} with group-probability vs. bit-probability decoding transmitted over an \ac{BI-AWGN}.}
    \label{fig:endogenous-result}
\end{figure*}

\subsection{Exogenous Correlation}
Next, we demonstrate how \emph{exogenous correlation} from correlated channels can improve decoding performance for non-binary \ac{TPC}. For this experiment, we choose a \ac{TPC} based on a \((16, 14)\) \ac{RS} component code with \(\bitcell=4\) bits per group. The \acp{TPC} are decoded with at max \(I_\text{max} = 10\) full decoding iterations.
The bits of each group are assigned either to \(\symcell=2\) 4-QAM or 4-ASK symbols. As discussed in Sec.~\ref{sec:theory}, the 4-ASK modulation theoretically achieves the same \ac{SNR} gains as the 16-QAM when transmitted over a correlated channel due to the independence of the I- and Q-component. This property allows us to effectively analyze 4-QAM and 16-QAM transmission using the same \ac{TPC}. In practice, a 16-QAM can be protected by the given \ac{TPC} by assigning the I- and Q-component of each symbol to different cells.

\begin{figure}
    \centering
    \includegraphics[clip, trim=0pt 16pt 0pt 0pt]{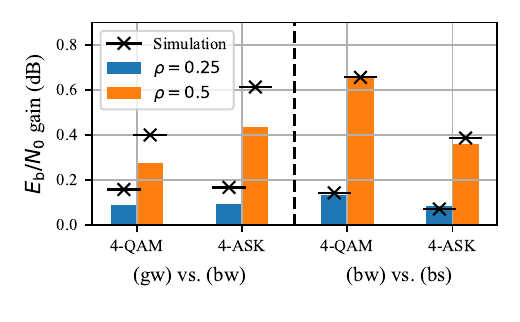}
    \caption{Non-binary \ac{TPC} decoding with exogenous correlation: The plot compares the calculated vs. simulated gains for the decoding of a non-binary \ac{TPC} based on an RS \((16, 14)\) code. The \ac{TPC} is either transmitted as 4-QAM or 4-ASK symbols, where each code cell is assigned \(\symcell=2\) consecutive modulation symbols.}
    \label{fig:tpc-snr-gain}
\end{figure}

\begin{figure*}
    \centering
    \subfloat[4-QAM Modulation]{
        \includegraphics[height=0.23\linewidth]{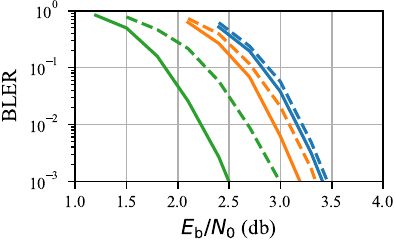}
        \label{fig:second_figure}
    }
    \hfill
    \subfloat[4-ASK Modulation]{
        \includegraphics[height=0.23\linewidth]{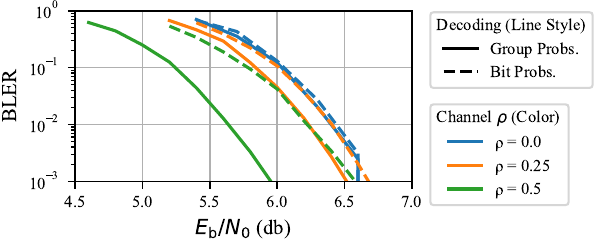}
        \label{fig:first_figure}
    }
    \caption{Non-binary \ac{TPC} decoding with exogenous correlation: The plot shows \ac{BLER} performance of a non-binary \ac{TPC} based on an RS (16, 14) component code transmitted over a Gauss-Markov channel with either a 4-QAM or 4-ASK modulation. The Gauss-Markov channel introduces a correlation between the received symbols parametrized by \(\rho\). The \ac{TPC} is either decoded with bit-probability (dashed) or group-probability decoding (solid). Group probabilities capture part of the correlation between coded bits, which improves decoding performance compared to correlation-free bit probabilities.}
    \label{fig:tpc-exogenous-correlation}
\end{figure*}

Figure~\ref{fig:tpc-exogenous-correlation} shows the \ac{BLER} performance of the non-binary \ac{TPC} with either group-probability or bit-probability decoding.
The channel probabilities of both decoding schemes calculate the windowed \ac{APP} of the two modulation symbols of each cell (i.e., preprocessing scheme (gw) and (bw)).
Group probability decoding outperforms bit-probability decoding by up to \SI{0.6}{\dB} even though the \ac{APP} estimator of both schemes exploits the correlation of the same window of modulation symbols. This result demonstrates the advantage of non-binary \ac{TPC} for correlated channels: it achieves additional performance gains that are not possible with a binary code while maintaining the simplicity of a \ac{SDD} scheme.

Figure~\ref{fig:tpc-snr-gain} shows the gains in \(\lEbNO\) at a target \ac{BLER} of \num{e-3} and compares them with the theoretical results in Fig.~\ref{fig:snr-gain}.
In addition to the comparison between group and bit-probability decoding with processing (gw) and (bw), it also compares the performance of bit-probability decoding of the windowed \ac{APP} estimator output (bw) with the bit-probability decoding with the linear equalizer output (bs), which ignores the correlation of the noise samples \(Z_i\) in \eqref{eqn:channel-model-gm} completely.
As predicted, the gains of group-probability decoding increase with increasing modulation order.

\section{Conclusion}
From a theoretical and practical perspective, we demonstrated the advantages of group probabilities over bit probabilities in turbo product decoding. Our theoretical analysis reveals that group-probability decoding retains critical correlations between codeword bits that are otherwise lost in traditional bit-probability approaches. These include exogenous correlation introduced by the channel and endogenous correlation introduced by a component decoder during iterative decoding.

We revisited non-binary \acp{TPC} as a practical code structure for group-probability decoding and to validate our theoretical results. Simulation results confirm the key findings of our analysis: for exogenous correlation, gains increase with modulation order and reach up to \SI{0.3}{\dB}; for endogenous correlation, gains increase with decreasing component code rate and reach up to \SI{0.7}{\dB}. These developments are enabled by the advent of new \ac{SISO} decoder such as \ac{SOGRAND}, which open up a broader class of codes that have been previously unexplored. In particular, the recently developed non-binary ORBGRAND chip~\cite{Kizilates25ORBRGANDAI} demonstrates the practicality of group-probability decoding for non-binary \acp{TPC}. Given that \acp{TPC} naturally support low-rate component codes, they are a compelling candidate for further exploration in low-rate applications.

{\appendices
\section{Derivation Gauss-Markov Model}\label{sec:gauss-markov-model-derivation}
To derive the linear equalizer, we apply the Z-transform to \eqref{eqn:channel-model-isi}, resulting in
\[
    \ztrafo{\widetilde{Y}} = \ztrafo{X} (1-\rho z^{-1}) + \ztrafo{\widetilde{Z}}.
\]
The \ac{ISI} can be compensated with the filter \(F(z) \coloneqq (1-\rho z^{-1})^{-1}\):
\[
    \ztrafo{Y} = F(z) \ztrafo{\widetilde{Y}}
    =
    \ztrafo{X} + F(z) \ztrafo{\widetilde{Z}}.
\]
By transforming back to the time domain, we obtain \cite{oppenheimDiscretetimeSignalProcessing1999}
\[
    Y_i = X_i + \sum_{j=0}^\infty \widetilde{Z}_{i-j} \rho^j = X_i + Z_i,
\]
where \(Z_i = \sum_{j=0}^\infty \widetilde{Z}_{i-j} \rho^j\).
Since I- and Q-component of \(\widetilde{Z}_i\) are independent, they are also independent for \(Z_i\).
The I- and Q-component of the sequence \(Z\) are each Gaussian processes with zero-mean and auto covariance 
\begin{align*}
    &E\{Z_{\text{I},i_1} Z_{\text{I},i_2}\} 
    = E\{Z_{\text{Q},i_1} Z_{\text{Q},i_2}\} \\
    &=
    \sum_{j_1 = 0}^\infty \sum_{j_2 = 0}^\infty 
    \underbrace{E\{ \widetilde{Z}_{\text{I},i_1-j_1} \widetilde{Z}_{\text{I},i_2-j_2} \}}_{= \widetilde{\sigma}^2 \delta_{i-j_1-(k-j_2)}} 
    \rho^{j_1} \rho^{j_2}
    =
    \underbrace{\frac{\widetilde{\sigma}^2}{1- \rho^2}}_{\eqqcolon \sigma^2}
    \rho^{|i_1-i_2|}
\end{align*}
In the last step, we used that the noise samples \(\widetilde{Z}_{\text{I},i}\) are independent for different times \(i\) and the geometric series. The independence results in a Kronecker delta \(\delta_{\tau}\), which is one if \(\tau=0\), and zero otherwise.

\section{Proof of Theorem~\ref{thm:information-rate-gm}}\label{sec:proof-theorem}
\begin{proof}
    As established in Sec.~\ref{sec:preprocessing}, the group \(Z_{\gstart}^{\gend}\) is multivariante normal distributed with covariance matrix \(C_{ij} = \sigma^2 \rho^{|i-j|}\) for \(i, j \in [\symcell]\). The differential entropy of this distribution is \cite{duffyUsingChannelCorrelation2023a}
    \[
        h(Z_s^e) =
        \frac{1}{2} \log\left(
            (2 \pi e)^\symcell (1-\rho^2)^{\symcell - 1}
        \right).
    \]
    \(h(Y_s^e)\) is equal to \(-E_{Y_s^e}(\log(p_{Y_s^e}(Y_s^e)))\) by definition. Using the law of total probability, we get
    \[
        p_{Y_s^e}(Y_s^e)
        = \frac{1}{|\mathcal{X}|^\symcell} \sum_{x^\symcell \in \mathcal{X}^\symcell} p_{Z_{\gstart}^{\gend}}(Y_s^e - x^\symcell),
    \]
    and the term for \(h(Y_s^e)\) in the theorem follows directly.
    Likewise, 
    \begin{align*}
        h(Y_s^e | X_i) 
        &= -E_{Y_s^e, X_i} \left\{
            \log(p_{Y_s^e | X_i}(Y_s^e | X_i))
        \right\},
    \end{align*}
    where
    \begin{align*}
        p_{Y_s^e | X_i}(Y_s^e | X_i) 
        &= P_{X_i | Y_s^e}(X_i | Y_s^e) p_{Y_s^e}(Y_s^e) / P_{X_i}(X_i) \\
        &= 
        \sum_{\mathclap{t^\symcell \in \mathcal{X}^\symcell: t_i = X_i}}
        P_{X_s^e | Y_s^e}(t^\symcell | Y_s^e) p_{Y_s^e}(Y_s^e) / P_{X_i}(X_i) \\
        &= \frac{1}{|\mathcal{X}|^{\symcell-1}} \sum_{t^\symcell \in \mathcal{X}^\symcell: t_i = X_i} p_{Z_{\gstart}^{\gend}}(Y_s^e - t^\symcell). \qedhere
    \end{align*} 
\end{proof}

\end{document}